\documentclass[letterpaper, 10pt, conference]{ieeeconf}
\IEEEoverridecommandlockouts
\usepackage{graphicx, enumerate}
\usepackage{cite,array, bm}
\usepackage{dsfont, epstopdf, color}
\usepackage{mathrsfs, subfigure}
\usepackage{amssymb}
\usepackage[cmex10]{amsmath}
\usepackage{subeqnarray}
\usepackage{cases, stmaryrd}
\usepackage[noend]{algpseudocode}
\usepackage{algorithmicx,algorithm}

\usepackage{xcolor}

\newtheorem{theorem}{Theorem}

\newtheorem{definition}{Definition}
\newtheorem{remark}{Remark}

\DeclareMathOperator{\Id}{Id}
\DeclareMathOperator{\ess}{ess}

\DeclareMathOperator{\ini}{in}
\DeclareMathOperator{\out}{out}

\DeclareMathOperator{\enab}{enab}

\begin{document}

\title{\LARGE \bf Dynamic Quantization based Symbolic Abstractions for Nonlinear Control Systems
\thanks{This work was supported by the H2020 ERC Starting Grant BUCOPHSYS, the EU H2020 Co4Robots Project, the Swedish Foundation for Strategic Research (SSF), the Swedish Research Council (VR) and the Knut och Alice Wallenberg Foundation (KAW).}
}

\author{Wei~Ren and Dimos V. Dimarogonas
\thanks{W. Ren and D. Dimarogonas are with Division of Decision and Control Systems, EECS, KTH Royal Institute of Technology, SE-10044, Stockholm, Sweden.
Email: \texttt{\small weire@kth.se}, \texttt{\small dimos@kth.se}.}
}

\maketitle

\begin{abstract}
This paper studies the construction of dynamic symbolic abstractions for nonlinear control systems via dynamic quantization. Since computational complexity is a fundamental problem in the use of discrete abstractions, a dynamic quantizer with a time-varying quantization parameter is first applied to deal with this problem. Due to the dynamic quantizer, a dynamic approximation approach is proposed for the state and input sets. Based on the dynamic approximation, dynamic symbolic abstractions are constructed for nonlinear control systems, and an approximate bisimulation relation is guaranteed for the original system and the constructed dynamic symbolic abstraction. Finally, the obtained results are illustrated through a numerical example from path planning of mobile robots.
\end{abstract}

\section{Introduction}
\label{sec-intro}

The use of discrete abstractions \cite{Milner1989communication, Tabuada2006linear} has gradually become a standard approach for the design of hybrid systems due to two main advantages. First, because of discrete abstractions of continuous dynamics, one can deal with controller synthesis problems efficiently via techniques developed in the fields of supervisory control \cite{Ramadge1987supervisory} or algorithmic game theory \cite{Ramadge1987modular}. Second, with an inclusion or equivalence relationship between the original system and the discrete abstraction, the synthesized controller is guaranteed to be correct by design, and thus formal verification is either not needed or can be reduced \cite{Girard2012controller}. To construct discrete abstractions, it is essential to find an equivalence relation on the state set of dynamic systems. Such equivalence relation leads to a new system, which is on the quotient space and shares the properties of interest of the original system.

In the literature on the construction of discrete abstractions, the commonly-used approach is based on (bi)simulation relation and its variants in \cite{Pola2008approximately, Girard2010approximately, Girard2007approximation}. The (bi)simulation relation and its variants lead to equivalences of dynamic systems in an exact or approximate setting. On the other hand, because of time-invariant quantization regions and the resulting simple structures \cite{Ren2018quantized}, static quantizers are applied in the construction of discrete abstractions \cite{Pola2008approximately, Girard2010approximately}. For instance, the uniform quantizer, which is a simple static quantizer \cite{Delchamps1990stabilizing}, is used extensively in the construction of discrete abstractions. Since the uniform quantizer partitions the state set with equal distance, a huge computational complexity may occur in the construction of discrete abstractions \cite{Pola2008approximately, Girard2010approximately}. To reduce the computational complexity, dynamic quantizers can be instead applied such that the state and input sets can be partitioned locally with different distances, which is the main motivation of this paper.

This paper studies symbolic abstraction for nonlinear control systems via dynamic quantization. To this end, we apply a dynamic quantizer to approximate the state and input sets, then construct dynamic abstraction for nonlinear control systems, and establish an approximate bisimulation relation between the original system and the constructed abstraction. The main contributions of this paper are two-fold. To begin with, the dynamic quantizer is first applied to approximate the state and input sets. The applied dynamic quantizer is the zoom quantizer \cite{Liberzon2003hybrid}, which has a quantization parameter to adjust quantization level dynamically. Since the zoom quantizer is only applied for bounded regions, we propose a discrete-event based choosing strategy to generate new bounded region and update the quantization parameter. Hence, a dynamic approximation approach is proposed for the state and input sets, and only a sequence of bounded regions needs to be approximated. Compared with the static approximation in \cite{Zamani2012symbolic, Girard2010approximately, Pola2008approximately}, the computational complexity is reduced greatly. Second, using the proposed dynamic approximation, a dynamic abstraction is constructed for nonlinear control systems. We further establish the approximate bisimulation relation between the original system and the constructed dynamic abstraction.

\section{Control Systems and Stability Property}
\label{sec-nonconsys}

\subsection{Notations}

$\mathbb{R}:=(-\infty, +\infty)$; $\mathbb{R}^{+}_{0}:=[0, +\infty)$; $\mathbb{R}^{+}:=(0, +\infty)$; $\mathbb{N}:=\{0, 1, \ldots\}$; $\mathbb{N}^{+}:=\{1, 2, \ldots\}$. $\mathbb{R}^{n}$ denotes the $n$-dimensional Euclidean space. Given two sets $\mathcal{A}$ and $\mathcal{B}$, $\mathcal{B}\backslash\mathcal{A}:=\{x: x\in\mathcal{B}, x\notin\mathcal{A}\}$. Given $a, b\in\mathbb{R}\cup\{\pm\infty\}$ with $a\leq b$, we denote by $[a, b]$ a closed interval. Given a vector $x\in\mathbb{R}^{n}$, $x_{i}$ denotes the $i$-th element of $x$, , $|x_{i}|$ denotes the absolute value of $x_{i}$, and $\|x\|$ denotes the infinite norm of $x$. Given a matrix $A\in\mathbb{R}^{n\times n}$, we denote by $\|A\|$ the infinity norm of $A$. The closed ball centered at $x\in\mathbb{R}^{n}$ with radius $\varepsilon\in\mathbb{R}^{+}$ is defined by $\mathbf{B}(x, \varepsilon)=\{y\in\mathbb{R}^{n}: \|x-y\|\leq\varepsilon\}$.

Given a measurable function $f: \mathbb{R}^{+}_{0}\rightarrow\mathbb{R}^{n}$, the (essential) supremum (sup norm) of $f$ is denoted by $\|f\|_{\infty}$; $\|f\|_{\infty}:=\ess\sup\{\|f(t)\|: t\in\mathbb{R}^{+}_{0}\}$; $f(t^{+})=\lim_{s\searrow t}f(s)$; $f^{+}=f(t^{+})$ when the time argument is omitted. A function $\alpha: \mathbb{R}^{+}_{0}\to\mathbb{R}^{+}_{0}$ is of class $\mathcal{K}$ if it is continuous, zero at zero, and strictly increasing; $\alpha(t)$ is of class $\mathcal{K}_{\infty}$ if it is of class $\mathcal{K}$ and unbounded. A function $\beta: \mathbb{R}^{+}_{0}\times\mathbb{R}^{+}_{0}\to\mathbb{R}^{+}_{0}$ is of class $\mathcal{KL}$ if $\beta(s, t)$ is of class $\mathcal{K}$ for each fixed $t\in\mathbb{R}^{+}_{0}$ and decreases to zero as $t\rightarrow\infty$ for each fixed $s\in\mathbb{R}^{+}_{0}$. $\Id$ denotes the identity function, and $\Id_{X}$ denotes the identity function from the set $X$ to $X$. Given two sets $A, B\subset\mathbb{R}^{n}$, a relation $\mathcal{R}\subset A\times B$ is a map $\mathcal{R}: A\rightarrow2^{B}$ defined by $b\in\mathcal{R}(a)$ if and only if $(a, b)\in\mathcal{R}$. $\mathcal{R}^{-1}$ denotes the inverse relation of $\mathcal{R}$, i.e., $\mathcal{R}^{-1}:=\{(b, a)\in B\times A: (a, b)\in\mathcal{R}\}$.

\subsection{Nonlinear Control Systems}

\begin{definition}[see \cite{Pola2008approximately}]
\label{def-1}
A \textit{control system} $\Sigma$ is a quadruple $\Sigma=(\mathbb{R}^{n}, U, \mathcal{U}, f)$, where, (i) $\mathbb{R}^{n}$ is the state set; (ii) $U\subseteq\mathbb{R}^{m}$ is the input set; (iii) $\mathcal{U}$ is a subset of all piecewise continuous functions of time from the interval $(a, b)\subset\mathbb{R}$ to $U$, with $b>0>a$; (iv) $f: \mathbb{R}^{n}\times U\rightarrow\mathbb{R}^{n}$ is a continuous map satisfying the following Lipschitz assumption: there exists a constant $L\in\mathbb{R}^{+}$ such that for all $x, y\in\mathbb{R}^{n}$ and all $u\in U$, we have $\|f(x, u)-f(y, u)\|\leq L\|x-y\|$.
\end{definition}

A curve $\xi: (a, b)\rightarrow\mathbb{R}^{n}$ is said to be a \textit{trajectory} of $\Sigma$ if there exists $\mathbf{u}\in\mathcal{U}$ such that, for almost all $t\in(a, b)$,
\begin{equation}
\label{eqn-1}
\dot{\xi}(t)=f(\xi(t), \mathbf{u}(t)).
\end{equation}
Different from the trajectory defined above over the open domain, we refer to the trajectory $\mathbf{x}: [0, \tau]\rightarrow\mathbb{R}^{n}$ defined on a closed domain $[0, \tau]$ with $\tau\in\mathbb{R}^{+}$ such that $\mathbf{x}=\xi|_{[0, \tau]}$. Denote by $\mathbf{x}(t, x, \mathbf{u})$ the point reached at time $t\in(a, b)$ under the input $\mathbf{u}$ from the initial state $x$. Such a point is determined uniquely because the assumptions on $f$ ensure the existence and uniqueness of the trajectory; see \cite[Appendix C.3]{Sontag2013mathematical}. The system $\Sigma$ is said to be \emph{forward complete} if every trajectory is defined on an interval of the form $(a, +\infty)$. Sufficient and necessary conditions can be found in \cite{Angeli1999forward} for the forward completeness of a control system.

\begin{definition}
\label{def-2}
The  system $\Sigma$ is \textit{incrementally globally asymptotically stable ($\delta$-GAS)}, if there exists $\beta\in\mathcal{KL}$ such that for all $t\in\mathbb{R}^{+}_{0}$, $x_{1}, x_{2}\in\mathbb{R}^{n}$ and all $\mathbf{u}\in\mathcal{U}$,
\begin{align}
\label{eqn-2}
\|\mathbf{x}(t, x_{1}, \mathbf{u})-\mathbf{x}(t, x_{2}, \mathbf{u})\|\leq\beta(\|x_{1}-x_{2}\|, t).
\end{align}
\end{definition}

Since it is not easy to check \eqref{eqn-2} directly, we can characterize the $\delta$-GAS property via the Lyapunov function.

\begin{definition}[see \cite{Girard2010approximately}]
\label{def-3}
A smooth function $V: \mathbb{R}^{n}\times\mathbb{R}^{n}\rightarrow\mathbb{R}^{+}_{0}$ is called a \textit{$\delta$-GAS Lyapunov function} for the system $\Sigma$, if there exist $\alpha_{1}, \alpha_{2}, \rho\in\mathcal{K}_{\infty}$ such that for all $x_{1}, x_{2}\in\mathbb{R}^{n}$ and $u\in U$,
\begin{align*}
&\alpha_{1}(\|x_{1}-x_{2}\|)\leq V(x_{1}, x_{2})\leq\alpha_{2}(\|x_{1}-x_{2}\|), \\
&\dfrac{\partial V(x_{1}, x_{2})}{\partial x_{1}}f(x_{1}, u)+\dfrac{\partial V(x_{1}, x_{2})}{\partial x_{2}}f(x_{2}, u) \nonumber \\
&\quad \leq-\rho(V(x_{1}, x_{2})).
\end{align*}
\end{definition}

From \cite{Zamani2015symbolic, Girard2010approximately}, for a compact set, the system $\Sigma$ is $\delta$-GAS if and only if it admits a $\delta$-GAS Lyapunov function.

\section{Approximate Equivalence Notions}
\label{sec-approbisimu}

In this section, we recall the notion of approximate (bi)simulation relation.

\begin{definition}[see \cite{Girard2007approximation}]
\label{def-4}
A \textit{transition system} is a sextuple $T = (X, X^{0}, U, \Delta, Y, H)$, consisting of: (i) a set of states $X$; (ii) a set of initial states $X^{0}\subseteq X$; (iii) a set of inputs $U$; (iv) a transition relation $\Delta\subseteq X\times U\times X$; (v) a set of outputs $Y$; (vi) an output function $H: X\rightarrow Y$. The transition system $T$ is said to be \textit{metric} if the output set $Y$ is equipped with a metric $\mathbf{d}: Y\times Y\rightarrow\mathbb{R}^{+}_{0}$, and \textit{symbolic} if the sets $X$ and $U$ are finite or countable.
\end{definition}

The transition $(x, u, x')\in\Delta$ is denoted by $x'\in\Delta(x, u)$, which means that the system can evolve from a state $x$ to a state $x'$ under the input $u$. An input $u\in U$ belongs to \textit{the set of enabled inputs} at the state $x$, denoted by $\enab(x)$, if $\Delta(x, u)\neq\varnothing$. If $\enab(x)=\varnothing$, then $x$ is said to be \textit{blocking}, otherwise, it is said to be \textit{non-blocking}. The transition system $T$ is said to be \textit{deterministic}, if $\Delta(x, u)$ has exactly one element for all $x\in X$ and all $u\in\enab(x)$. In this case, we write $x'=\Delta(x, u)$ with a slight abuse of notation.

\begin{definition}[see \cite{Girard2007approximation}]
\label{def-5}
Let $T_{i}=(X_{i}, X^{0}_{i}, U, \Delta_{i}, Y, H_{i})$, $i=1, 2$, be two metric transition systems with the same input set $U$ and output set $Y$ equipped with the metric $\mathbf{d}$. Let  $\varepsilon\in\mathbb{R}^{+}_{0}$, a relation $\mathcal{R}\subseteq X_{1}\times X_{2}$ is said to be an \textit{$\varepsilon$-approximate simulation relation} from $T_{1}$ to $T_{2}$, if for all $(x_{1}, x_{2})\in\mathcal{R}$ and all $u\in U$: (i) $\mathbf{d}(H_{1}(x_{1}), H_{2}(x_{2}))\leq\varepsilon$; (ii) for each $x'_{1}\in\Delta_{1}(x_{1}, u)$, there exists $x'_{2}\in\Delta_{2}(x_{2}, u)$ such that $(x'_{1}, x'_{2})\in\mathcal{R}$. A relation $\mathcal{R}\subseteq X_{1}\times X_{2}$ is said to be an \textit{$\varepsilon$-approximate bisimulation relation} between $T_{1}$ and $T_{2}$, if for all $(x_{1}, x_{2})\in\mathcal{R}$ and all $u\in U$, the conditions (i)-(ii) hold, and (iii) for each $x'_{2}\in\Delta_{2}(x_{2}, u)$, there exists $x'_{1}\in\Delta_{1}(x_{1}, u)$ such that $(x'_{1}, x'_{2})\in\mathcal{R}$.
\end{definition}

Denote by $T_{1}\preceq^{\varepsilon}_{\mathcal{S}}T_{2}$ if there is an $\varepsilon$-approximate simulation relation from $T_{1}$ to $T_{2}$. Denote by $T_{1}\simeq_{\varepsilon}T_{2}$ if there is an $\varepsilon$-approximate bisimulation relation $\mathcal{S}$ between $T_{1}$ and $T_{2}$ such that $\mathcal{R}(X_{1})=X_{2}$ and $\mathcal{R}^{-1}(X_{2})=X_{1}$.

\section{Dynamic Approximation}
\label{sec-dynappro}

In this section, a dynamic quantization based approximation approach is proposed for the state and input sets. To this end, we work on the time-discretization of the system $\Sigma$ with the sampling period $\tau>0$ as a design parameter. The sampled-data system can be written as a transition system $T_{\tau}(\Sigma):=(X_{1}, X^{0}_{1}, U_{1}, \Delta_{1}, Y_{1}, H_{1})$, where,
\begin{itemize}
\item the state set is $X_{1}:=\mathbb{R}^{n}$;
\item the set of initial states is $X^{0}_{1}:=\mathbb{R}^{n}$;
\item the input set is $U_{1}=\{u\in\mathcal{U}: \mathbf{x}(t, x, u) \text{ is defined for all } x\in\mathbb{R}^{n}\}$;
\item the transition relation is given by: for $x\in X_{1}$ and $u\in U_{1}$, $x'=\Delta_{1}(x, u)$ if and only if $x'=\mathbf{x}(\tau, x, u)$;
\item the output set is $Y_{1}:=\mathbb{R}^{n}$;
\item the output map is $H_{1}=\Id_{X_{1}}$.
\end{itemize}
$T_{\tau}(\Sigma)$ is non-blocking, deterministic, and metric when $Y_{1}$ is equipped with the metric $\textbf{d}(y, y')=\|y-y'\|$ for $y, y'\in Y_{1}$.

\subsection{Zoom Quantizer}
\label{subsec-dynquan}

A static quantizer is a piecewise constant function $\mathbf{q}: \mathbb{R}^{n}\rightarrow\mathcal{Q}$, where $\mathcal{Q}$ is a finite subset of $\mathbb{R}^{n}$; see \cite{Liberzon2003hybrid}. That is, the quantizer divides $\mathbb{R}^{n}$ into a finite number of quantization regions of the form $\{z\in\mathbb{R}^{n}: \mathbf{q}(z)=\jmath\in\mathcal{Q}\}$. For the quantizer $\mathbf{q}$, assume that there exist $M>\Lambda>0$ and $\Lambda_{0}>0$ such that the following conditions are satisfied \cite{Ren2018quantized, Liberzon2003hybrid}: (i) $\|z\|\leq M\Rightarrow\|\mathbf{q}(z)-z\|\leq\Lambda$; (ii) $\|z\|>M\Rightarrow\|\mathbf{q}(z)\|>M-\Lambda$; (iii) $\|z\|\leq\Lambda_{0}\Rightarrow\mathbf{q}(z)\equiv0$.
$M$ is called the range of the quantizer, and $\Lambda$ is called the upper bound of the quantization error $\mathbf{q}(z)-z$. Item (i) implies that if the signal does not saturate, then the quantization error is bounded by $\Lambda$. Item (ii) provides an approach to detecting whether the signal saturates or not. Item (iii) implies that it is reasonable to quantize a signal as zero directly if such signal is small. The quantizers satisfying the items (i)-(iii) can be found in the existing works; see for instance \cite{Elia2001stabilization, Delchamps1990stabilizing, Coutinho2010input}.

Following the static quantizer $\mathbf{q}$, the applied dynamic quantizer is the zoom quantizer defined as follows.
\begin{align}
\label{eqn-6}
Q_{\mu}(z)&:=\mu\mathbf{q}(z/\mu) \nonumber \\
&=\left\{\begin{aligned}
&M\Lambda\mu,  && z\geq(M+0.5)\Lambda\mu; \\
&k\Lambda\mu, && (k-0.5)\Lambda\mu\leq z<(k+0.5)\Lambda\mu; \\
&-M\Lambda\mu,  && z<-(M+0.5)\Lambda\mu,
\end{aligned}\right.
\end{align}
where, $\mu\in\mathbb{R}^{+}$ is called the quantization parameter initialized at $\mu_{0}\in\mathbb{R}^{+}$; and $k\in\mathcal{M}:=\{-M, \ldots, M\}$. The quantization parameter $\mu$ is time-varying instead of constant. For instance, $\mu$ is constant in the transmission interval and updated at the transmission times for networked control systems; see \cite{Heemels2010networked}. For the zoom quantizer \eqref{eqn-6}, the quantization range is $M\Lambda\mu$ and the upper bound of the quantization error is $\Lambda\mu$. Both $M\Lambda\mu$ and $\Lambda\mu$ depend on the quantization parameter, which implies that the quantization regions of $Q_{\mu}$ are dynamic.

According to simple geometrical considerations, we can check that for all $z\in\mathbb{R}$, $\|z-Q_{\mu}(z)\|\leq\Lambda\mu$ if $\|z\|\leq(M+1)\Lambda\mu$ and $\|z-Q_{\mu}(z)\|>\Lambda\mu$ if $\|z\|>(M+1)\Lambda\mu$. That is, the quantization error is bounded in the quantization region $\{z\in\mathbb{R}: \|z\|\leq(M+1)\Lambda\mu\}$. If the system state escapes from the bounded quantization region, then the quantization error may be large enough to lead to the failure of the desired performances and specifications. In addition, even within the quantization region, if the quantization parameter $\mu$ is too large, then the quantization can be so coarse that the control strategy is not applicable or even does not exist. As a result, the larger the quantization parameter is, the larger the bounded quantization region is, the coarser the quantization is, which leads to a tradeoff between the bounded quantization regions and the approximation accuracy.

\subsection{Discrete Event based Choosing Strategy}
\label{subsec-eventupdate}

To apply the quantizer \eqref{eqn-6} with the bounded quantization region into the construction of symbolic abstractions, an approach is to choose finite bounded regions, whose sizes are not larger than those of the quantization regions, such that the state trajectory is covered by these chosen regions. That is, a new region is chosen before the system state escapes from the current chosen region. Hence, the union of these chosen regions is a cover or a subset of the state set.

Given an initial state $x_{0}\in X_{1}$, we can choose a bounded region $\mathcal{S}_{0}\subset X_{1}$ such that $x_{0}\in\mathcal{S}_{0}$, and thus $\mathcal{S}_{0}$ is set as the initial chosen region.
Next, a discrete-event based approach is proposed to choose the bounded regions starting from $\mathcal{S}_{0}$.

Given a bounded region $\mathcal{S}_{i}$, $i\in\mathbb{N}$, we introduce an auxiliary region given by $\mathcal{S}^{\omega}_{i}:=\{x\in\mathcal{S}_{i}|\mathbf{d}(x, x')\leq\omega\Rightarrow x'\in\mathcal{S}_{i}\}$, where $\omega\in(0, 1)$ is an parameter that depends on the desired precision. We can in fact choose $\omega$ as the desired precision for the sake of simplicity. $\mathcal{S}^{\omega}_{i}$ is called an $\omega$-contraction of $\mathcal{S}_{i}$. With the introduction of $\mathcal{S}^{\omega}_{i}$, the generation rule for the next region is given as follows: if and only if $x\in\mathcal{S}_{i}\backslash\mathcal{S}^{\omega}_{i}$, a new bounded region $\mathcal{S}_{i+1}$ is generated such that $x\in\mathcal{S}^{\omega}_{i+1}$. Since $x\in\mathcal{S}_{i}\backslash\mathcal{S}^{\omega}_{i}$ and $x\in\mathcal{S}^{\omega}_{i+1}$, we have that $\mathcal{S}_{i+1}\cap\mathcal{S}_{i}\neq\varnothing$. That is, the $\omega$-contraction of the new region is only required to intersect with the current one, and the reason lies in that the intersection $\mathcal{S}_{i}\cap\mathcal{S}_{i+1}$ provides an admissible region for the state trajectory in $\mathcal{S}_{i+1}$ to avoid the Zeno and chattering phenomena; see also \cite{Liberzon2003switching, Forstner2002discrete}. In addition, a discrete event $a\in\{0, 1\}$ is defined to measure whether the system state is in $\mathcal{S}_{i}\backslash\mathcal{S}^{\omega}_{i}$. That is, $a=1$ if $x\in\mathcal{S}_{i}\backslash\mathcal{S}^{\omega}_{i}$; otherwise, $a=0$.

\begin{remark}
\label{rmk-1}
A similar technique was applied in \cite{Forstner2002discrete}, where quantized control systems are modeled as a discrete-event system. In \cite{Liberzon2003switching, Forstner2002discrete}, $\mathcal{S}^{\omega}_{i}$ is the quantization region, whereas  $\mathcal{S}_{i}$ is the quantization region with additional border to measure the discrete events or to avoid chattering phenomena. Let $\omega:=(M+1)^{-1}$, then $\mathcal{S}^{\omega}_{i}:=\{z\in\mathbb{R}: \|z\|\leq M\Lambda\mu_{i}\}$ and $\mathcal{S}_{i}:=\{z\in\mathbb{R}: \|z\|\leq(M+1)\Lambda\mu_{i}\}$, which are the same as in \cite{Liberzon2003switching, Forstner2002discrete}, and can be a choosing strategy. Note that the quantization regions are not required to have any particular shapes, and thus the choosing strategy is not unique.
\hfill $\square$
\end{remark}

Since the choosing strategy of $\mathcal{S}_{i+1}$ is not restricted and just requires that $\mathcal{S}^{\omega}_{i+1}\cap\mathcal{S}_{i}\neq\varnothing$, the determination of the size of $\mathcal{S}_{i+1}$ has effects on the controller design, which in turn affects the achievement of the desired performances and specifications. In the following, for a given safety specification, we provide an update strategy to determine the size of the new region $\mathcal{S}_{i+1}$.

For many dynamic systems like embedded systems and multi-agent systems \cite{Hashimoto2017robust}, safety and reachability are two fundamental problems, which require the system to avoid the obstacle region $\mathcal{O}:=\bigcup_{l\in\mathfrak{L}}\mathcal{O}_{l}$ (where $\mathcal{O}_{l}$ is a single obstacle and $\mathfrak{L}$ is finite set) and to reach and then remain in a specific region; see \cite{Girard2013low, Girard2016safety, Girard2012controller}. Given a safety specification, there exists a safe region $\mathcal{O}_{s}$, which may consist of finite subregions. To capture the relations among the chosen regions $\mathcal{S}_{i}$, obstacle region $\mathcal{O}$, safe region $\mathcal{O}_{s}$, the following three discrete events are defined.
\begin{align*}
b&=\left\{\begin{aligned}
&0,  && \mathcal{S}_{i}\cap\mathcal{O}=\varnothing, \\
&1,  && \mathcal{S}_{i}\cap\mathcal{O}\neq\varnothing,
\end{aligned}\right. \quad
c=\left\{\begin{aligned}
&0,  && \mathcal{S}_{i}\cap\mathcal{O}_{s}=\varnothing, \\
&1,  && \mathcal{S}_{i}\cap\mathcal{O}_{s}\neq\varnothing.
\end{aligned}\right.
\end{align*}
Combining all the discrete-event variables yields an augmented variable $e:=(a, b, c)\in\mathbb{N}^{3}$. The variable $e\in\mathbb{N}^{3}$ is well-defined due to the choosing strategy, and updated with the generation of the new region.

The size of the new region $\mathcal{S}_{i+1}$ depends on $e\in\mathbb{N}^{3}$. Conversely, the update of $e$ also depends on the choice of the region $\mathcal{S}_{i+1}$. In the following, an update strategy is derived for both the size of the region $\mathcal{S}_{i+1}$ and the variable $e\in\mathbb{N}^{4}$.
\begin{itemize}
  \item If $\Delta_{1}(x, u)\in\mathcal{S}^{\omega}_{i}$, then no new region is generated and the variable $e$ is not updated, that is, $e^{+}=e$.
  \item If $\Delta_{1}(x, u)\in\mathcal{S}_{i}\backslash\mathcal{S}^{\omega}_{i}$, then $a^{+}=1$, and a new region $\mathcal{S}_{i+1}$ is to be generated such that $\Delta_{1}(x, u)\in\mathcal{S}_{i}\cap\mathcal{S}^{\omega}_{i+1}$. The size of $\mathcal{S}_{i+1}$ is decided by the strategy below.
  \begin{itemize}
    \item If $b=1$ or $c=1$, then $\mathcal{S}_{i+1}$ is contracted with the contraction of the quantization parameter, that is,
    \begin{equation}
    \label{eqn-8}
    \mu_{i+1}=\Omega_{\ini}\mu_{i},
    \end{equation}
    where $\Omega_{\ini}\in(0, 1)$ is a given constant.
    \item If $b=0$ and $c=0$, then $\mathcal{S}_{i+1}$ is expanded with the expansion of the quantization parameter, that is,
    \begin{equation}
    \label{eqn-9}
    \mu_{i+1}=\Omega_{\out}\mu_{i},
    \end{equation}
    where $\Omega_{\out}\geq1$ is a given constant.
  \end{itemize}
  Once $\mathcal{S}_{i+1}$ is determined, $e\in\mathbb{N}^{3}$ is updated based on the relations among the regions $\mathcal{S}_{i+1}$, $\mathcal{O}$ and $\mathcal{O}_{s}$.
\end{itemize}

In the above update strategy, there is no need to generate a new region and discrete events are not changed if $x\in\mathcal{S}^{\omega}_{i}$. If not, a new region is generated; $e$ and $\mu$ are updated. Note that the update of $\mu$ (i.e., $\mu_{i+1}$) is based on the current value $\mu_{i}$, and thus the size of the new region $\mathcal{S}_{i+1}$ is determined via $\mu_{i+1}$; see also Remark \ref{rmk-1}. The new region $\mathcal{S}_{i+1}$ is contracted if $b=1$ or $c=1$. $b=1$ means that the region $\mathcal{S}_{i}$ intersects with the obstacle region. Thus, the new region $\mathcal{S}_{i+1}$ is contracted to avoid the obstacle region. $c=1$ means that the region $\mathcal{S}_{i}$ intersects with the safe region, and thus the new region $\mathcal{S}_{i+1}$ is contracted to drive the system state to reach the safe region. On the other hand, the region $\mathcal{S}_{i+1}$ is expanded if $b=0$ and $c=0$. $b=c=0$ means that the region $\mathcal{S}_{i}$ is admissible for all the states and enabled inputs, and then it is reasonable to expand the new region to allow more admissible states and inputs. Such a setting is able to guarantee the satisfaction of the safety specification.

\subsection{Approximation of State and Input Sets}
\label{subsec-approximate}

According to the sequence of the regions chosen in Subsection \ref{subsec-eventupdate}, the approximation of the state and input sets is presented in this subsection. Here, we focus on the approximation of a chosen region $\mathcal{S}_{i}\subset X_{1}$, $i\in\mathbb{N}$.

For each region $\mathcal{S}_{i}$, there exists $\mu_{i}\in\mathbb{R}^{+}$. Using the zoom quantizer \eqref{eqn-6}, the region $\mathcal{S}_{i}$ is approximated by a finite sequence of embedded lattices $[\mathcal{S}_{i}]_{\mu_{i}}$, where
\begin{align*}
[\mathcal{S}_{i}]_{\mu_{i}}:=\left\{q\in\mathcal{S}_{i}: q_{l}=\frac{k_{l}\Lambda\mu_{i}}{\sqrt{n}}, k_{l}\in\mathcal{M}, l\in\{1, \ldots, n\}\right\}.
\end{align*}
Based on geometrical considerations, the region $\mathcal{S}_{i}$ can be covered by the set $\bigcup_{q\in[\mathcal{S}_{i}]_{\mu_{i}}}\mathbf{B}(q, \lambda_{i})$ with $\lambda_{i}\geq\Lambda\mu_{i}/2$. We associate a quantizer $Q_{\mu_{i}}: \mathcal{S}_{i}\rightarrow[\mathcal{S}_{i}]_{\mu_{i}}$ such that $Q_{\mu_{i}}(x)=q$ if and only if for $x=(x_{1}, \ldots, x_{n})\in\mathcal{S}_{i}$, $l\in\{1, \ldots, n\}$,
\begin{align*}
(k_{l}-0.5)\frac{\Lambda\mu_{i}}{\sqrt{n}}<x_{l}\leq(k_{l}+0.5)\frac{\Lambda\mu_{i}}{\sqrt{n}}, \quad k_{l}\in\mathcal{M}.
\end{align*}

In the following, the set $U_{1}(\mathcal{S}_{i})$ is approximated by $U_{2}(\mathcal{S}_{i}):=\bigcup_{q\in[\mathcal{S}_{i}]_{\mu_{i}}}U_{2}(q)$, where the set $U_{2}(q)$ captures all the inputs that can be applied at the state $q\in[\mathcal{S}_{i}]_{\mu_{i}}$. To show the reasonability of $U_{2}(\mathcal{S}_{i})$, we explain the definition of $U_{2}(q)$ in detail. The definition of $U_{2}(q)$ is based on the notion of reachable sets. Given any state $q\in[\mathcal{S}_{i}]_{\mu_{i}}$ (thus $q\in\mathcal{S}_{i}$), the reachable state set of $T_{\tau}(\Sigma)$ from $q$ is obtained by $\mathfrak{R}_{i}(\tau, q):=\left\{x'\in\mathcal{S}_{i}: \mathbf{x}(\tau, q, u)=x', u\in U_{1}(\mathcal{S}_{i})\right\}$. The set $\mathfrak{R}_{i}(\tau, q)$ is well defined due to the definition of $U_{1}(\mathcal{S}_{i})$.

The reachable set $\mathfrak{R}_{i}(\tau, q)$ is approximated as follows. Given any $\eta_{i}\in\mathbb{R}^{+}$, consider the following set $\mathcal{Z}_{\eta_{i}}(\tau, q):=\{y\in[\mathcal{S}_{i}]_{\eta_{i}}: \exists z\in\mathfrak{R}_{i}(\tau, q) \text{ such that }\|y-z\|\leq\Lambda\eta_{i}/2\}$. Here, the choice of $\eta_{i}$ is not related to $\mu_{i}$, and limited by the desired precision; see Section \ref{sec-symabs}. Define the function $\psi_{i}: \mathcal{Z}_{\eta_{i}}(\tau, q)\rightarrow U_{1}(\mathcal{S}_{i})$, which is defined such that for any $y\in\mathcal{Z}_{\eta_{i}}(\tau, q)$, there exists an input $u_{1}=\psi_{i}(y)\in U_{1}(\mathcal{S}_{i})$ such that $\|y-\mathbf{x}(\tau, q, u_{1})\|\leq\Lambda\eta_{i}/2$. Note that the function $\psi_{i}$ is not unique. The set $U_{2}(q)$ can be defined by $U_{2}(q):=\psi_{i}(\mathcal{Z}_{\eta_{i}}(\tau, q))$. Since the set $U_{2}(q)$ is the image through the map $\psi_{i}$ of a countable set, we have that $U_{2}(q)$ is countable, which implies that $U_{2}(\mathcal{S}_{i})$ is countable. Hence, the set $U_{2}(\mathcal{S}_{i})$ approximates the set $U_{1}(\mathcal{S}_{i})$ in the following way: given any $q\in[\mathcal{S}_{i}]_{\mu_{i}}$, for any $u_{1}\in U_{1}(\mathcal{S}_{i})$, there exists $u_{2}\in U_{2}(q)$ such that $\|\mathbf{x}(\tau, q, u_{1})-\mathbf{x}(\tau, q, u_{2})\|\leq\Lambda\eta_{i}$.


Using a similar mechanism, the approximation can be obtained for all the chosen regions. In the intersection regions, two approximations may overlap, and both are available for the proceeding construction of the symbolic abstraction.

\section{Symbolic Abstraction}
\label{sec-symabs}

Following Section \ref{sec-dynappro}, we are ready to construct a symbolic abstraction for $T_{\tau}(\Sigma)$. Since the union of all the bounded regions chosen in Subsection \ref{subsec-eventupdate} is a cover of the state trajectory, we only focus on the union of these chosen bounded regions. Define the region $\mathcal{S}:=\bigcup_{i\in\mathbb{N}^{+}}\mathcal{S}_{i}\subset\mathbb{R}^{n}$. That is, the state set is constrained in a subset of $\mathbb{R}^{n}$, which is the case in physical systems. As a result, the system $T_{\tau}(\Sigma)$ can be rewritten as $\bar{T}_{\tau}(\Sigma):=(\bar{X}_{1}, X^{0}_{1}, \bar{U}_{1}, \Delta_{1}, Y_{1}, H_{1})$ with $\bar{X}_{1}=\mathcal{S}$ and $\bar{U}_{1}=U_{1}(\mathcal{S})$. The region $\mathcal{S}$ is not fixed, but varies with the choosing strategy in Subsection \ref{subsec-eventupdate}.

Define the vectors $\bar{\mu}:=(\mu_{0}, \mu_{1}, \mu_{2}, \ldots)$ and $\bar{\eta}:=(\eta_{0}, \eta_{1}, \eta_{2}, \ldots)$. The symbolic abstraction of $\bar{T}_{\tau}(\Sigma)$ is described by the transition system
\begin{equation}
\label{eqn-11}
T_{\tau, \bar{\mu}, \bar{\eta}}(\Sigma)=(X_{2}, X^{0}_{2}, U_{2}, \Delta_{2}, Y_{2}, H_{2}),
\end{equation}
where,
\begin{itemize}
\item the set of states is $X_{2}=\bigcup_{i\in\mathbb{N}^{+}}[\mathcal{S}_{i}]_{\mu_{i}}$;
\item the set of initial states is $X^{0}_{2}=\mathbb{R}^{n}$;
\item the set of inputs is $U_{2}=\bigcup_{i\in\mathbb{N}^{+}}[U_{1}(\mathcal{S}_{i})]_{\eta_{i}}$;
\item the transition relation is given by: for $q\in X_{2}$ and $u\in U_{2}$, $q'=\Delta_{2}(q, u)$ if and only if $q'=Q_{\mu_{i}}(\mathbf{x}(\tau, q, u))$;
\item the set of outputs is $Y_{2}=\mathbb{R}^{n}$;
\item the output map is $H_{2}=\Id_{X_{2}}$.
\end{itemize}
$T_{\tau, \bar{\mu}, \bar{\eta}}(\Sigma)$ is non-blocking, deterministic. It is metric if $Y_{2}$ is equipped with the metric $\mathbf{d}(y', y)=\|y-y'\|$ for $y, y'\in Y_{2}$.

\begin{theorem}
\label{thm-1}
Consider the system $\Sigma$ and given a desired precision $\varepsilon\in\mathbb{R}^{+}$, if $\Sigma$ is $\delta$-GAS, and there exist $\tau, \eta_{i}, \mu_{i}\in\mathbb{R}^{+}$, $i\in\mathbb{N}$, such that
\begin{align}
\label{eqn-12}
\beta(\varepsilon, \tau)+\Lambda\eta_{i}+0.5\Lambda\mu_{i}\leq\varepsilon,
\end{align}
then $\bar{T}_{\tau}(\Sigma)\simeq_{\varepsilon}T_{\tau, \bar{\mu}, \bar{\eta}}(\Sigma)$.
 \end{theorem}

\begin{proof}
Define the relation $\mathcal{R}:=\{(x, q)\in\bar{X}_{1}\times X_{2}: \|x-q\|\leq\varepsilon\}$. We will show that $\mathcal{R}$ is an $\varepsilon$-approximate bisimulation relation between $\bar{T}_{\tau}(\Sigma)$ and $T_{\tau, \bar{\mu}, \bar{\eta}}(\Sigma)$.

From the relation $\mathcal{R}$, we have that $\mathcal{R}(\bar{X}_{1})=X_{2}$. In addition, $\bar{X}_{1}\subseteq\bigcup_{q\in X_{2}}\mathbf{B}(q, \Lambda\mu_{i}/2)$, which implies that for any $x\in\bar{X}_{1}$, there exists a $q\in X_{2}$ such that $\|x-q\|\leq\Lambda\mu_{i}/2$. From the relation $\mathcal{R}$,
$\Lambda\mu_{i}/2\leq\varepsilon$ holds for all $i\in\mathbb{N}^{+}$. Next, we only consider the region $\mathcal{S}_{i}$, where $i\in\mathbb{N}^{+}$.

Consider any $(x, q)\in\mathcal{R}$ with $x\in\mathcal{S}_{i}$ and given any $u_{1}\in U_{1}(\mathcal{S}_{i})$ such that $x'=\Delta_{1}(x, u_{1})\in\mathcal{S}^{\omega}_{i}$. Let  $\bar{x}=\mathbf{x}(\tau, q, u_{1})\in\mathcal{S}_{i}$. Since $\mathcal{S}_{i}\subseteq\bigcup_{q\in [\mathcal{S}_{i}]_{\eta_{i}}}\mathbf{B}(q, \Lambda\eta_{i}/2)$, there exists $v\in[\mathcal{S}_{i}]_{\eta_{i}}$ such that $\|\bar{x}-v\|=\|\mathbf{x}(\tau, q, u_{1})-v\|\leq\Lambda\eta_{i}/2$. Since $\bar{x}\in\mathfrak{R}_{i}(\tau, q)$, $v\in\mathcal{Z}_{\eta_{i}}(\tau, q)$ from the definition of $\mathcal{Z}_{\eta_{i}}(\tau, q)$. Given any $u_{2}\in U_{2}(\mathcal{S}_{i})$ with $u_{2}=\psi_{i}(v)$, there exists $w=\mathbf{x}(\tau, q, u_{2})\in\mathcal{S}_{i}$ such that $\|w-v\|=\|\mathbf{x}(\tau, q, u_{2})-v\|\leq\Lambda\eta_{i}/2$. In addition, since $\bar{X}_{1}\subseteq\bigcup_{q\in X_{2}}\mathbf{B}(q, \Lambda\mu_{i}/2)$, there exists $q'\in X_{2}$ such that $\|q'-w\|=\|q'-\mathbf{x}(\tau, q, u_{2})\|\leq\Lambda\mu_{i}/2$. Choose $q'=\Delta_{2}(q, u_{2})$. Since the system $\Sigma$ is $\delta$-GAS, we have
\begin{align*}
&\|x'-q'\|\leq\|\Delta_{1}(x, u_{1})-\bar{x}\|+\|\bar{x}-\Delta_{2}(q, u_{2})\| \\
&\leq\beta(\|x-q\|, \tau)+\|\bar{x}-v+v-w+w-\Delta_{2}(q, u_{2})\| \\
&\leq\beta(\varepsilon, \tau)+\Lambda\eta_{i}+\Lambda\mu_{i}/2.
\end{align*}
Therefore, we conclude from \eqref{eqn-12} that $(x', q')\in\mathcal{R}$. If $x'=\Delta_{1}(x, u_{1})\in\mathcal{S}_{i}\backslash\mathcal{S}^{\omega}_{i}$, then it follows from the choosing strategy in Section \ref{sec-dynappro} that $x'\in\mathcal{S}_{i+1}$. In this case, since \eqref{eqn-12} holds for all $i\in\mathbb{N}$, we can proceed the analysis along the similar fashion, and obtain that $(x', q')\in\mathcal{R}$.

Consider any $(x, q)\in\mathcal{R}$ with $x\in\mathcal{S}_{i}$. Choose any $u_{2}\in U_{2}(\mathcal{S}_{i})$ such that $q'=\Delta_{2}(q, u_{2})\in\mathcal{S}_{i}$. From the construction of $T_{\tau, \bar{\mu}, \bar{\eta}}(\Sigma)$, we have that $q'=Q_{\mu_{i}}(\mathbf{x}(\tau, q, u))$, which implies that $\|q'-\mathbf{x}(\tau, q, u_{2})\|\leq\Lambda\mu_{i}/2$. Since $\mathbf{x}(\tau, q, u_{2})=\Delta_{1}(q, u_{2})$, we have that $\mathbf{x}(\tau, q, u_{2})\in\bar{X}_{1}$. Pick $u_{1}=u_{2}$, and let $x'=\Delta_{1}(x, u_{1})\in\mathcal{S}_{i}$. It follows from $\delta$-GAS of $\Sigma$ that
\begin{align*}
\|x'-q'\|&\leq\|x'-\mathbf{x}(\tau, q, u_{2})\|+\|\mathbf{x}(\tau, q, u_{2})-q'\| \\
&\leq\beta(\|x-q\|, \tau)+\|\mathbf{x}(\tau, q, u_{2})-q'\| \\
&\leq\beta(\varepsilon, \tau)+\Lambda\mu_{i}/2,
\end{align*}
which implies from \eqref{eqn-12} that $(x', q')\in\mathcal{R}$. Similarly, if $q'\in\mathcal{S}_{i+1}$, then the analysis can be continued in a similar fashion to obtain that $(x', q')\in\mathcal{R}$.
\end{proof}

Since the quantization parameter $\mu$ evolves with the rules given in \eqref{eqn-8}-\eqref{eqn-9}, the value of $\mu$ in each chosen region is only related to the initial $\mu_{0}$ and the coefficients $\Omega_{\ini}$ and $\Omega_{\out}$. In addition, the update of $\mu$ corresponds to the choice of the new region. Set $\mathcal{S}_{0}$ as the first chosen region, then $\mathcal{S}_{i}$ is the $(i+1)$-th chosen region with $i\in\mathbb{N}$. In these $i+1$ times of choosing regions, the number of the updates of $\mu$ via \eqref{eqn-8}-\eqref{eqn-9} is $i$. Denote by $p\in\mathbb{N}$ the number of the updates of $\mu$ via \eqref{eqn-8}, and thus the number of the updates of $\mu$ via \eqref{eqn-9} is $i-p\in\mathbb{N}$.

\begin{theorem}
\label{thm-2}
Consider the system $\Sigma$ and given any desired precision $\varepsilon\in\mathbb{R}^{+}$. If $\Sigma$ is $\delta$-GAS, and there exist $\tau, \eta_{0}, \mu_{0}\in\mathbb{R}^{+}$, $\Omega_{\ini}\in(0, 1)$ and $\Omega_{\out}>1$ such that for all $i\in\mathbb{N}$,
\begin{align}
\label{eqn-18}
\beta(\varepsilon, \tau)+\Omega^{p}_{\ini}\Omega^{i-p}_{\out}(\Lambda\eta_{0}+\Lambda\mu_{0}/2)\leq\varepsilon, \quad p\leq i,
\end{align}
then $\bar{T}_{\tau}(\Sigma)\simeq_{\varepsilon}T_{\tau, \bar{\mu}, \bar{\eta}}(\Sigma)$.
\end{theorem}

Let $\mu_{i}:=\Omega^{p}_{\ini}\Omega^{i-p}_{\out}\mu_{0}$ and $\eta_{i}:=\Omega^{p}_{\ini}\Omega^{i-p}_{\out}\eta_{0}$ for $i\in\mathbb{N}^{+}$, then Theorem \ref{thm-2} is equivalent to Theorem \ref{thm-1}. Condition \eqref{eqn-18} provides an alternative criterion for the choice of the bounded regions: the update of $\mu$ is required to satisfy \eqref{eqn-18}.

According to the discussion in \cite[Section 4]{Saoud2017optimal}, given a compact set $C\subset\mathbb{R}^{n}$ with nonempty interior, the number of symbolic states in $X_{2}\cap C$ is $\theta_{C}\eta^{-n}$, where $\theta_{C}\in\mathbb{R}^{+}$ is a positive constant proportional to the volume of $C$, and $\eta$ is the state space sampling parameter. Comparing with \cite{Girard2013low, Girard2016safety, Pola2008approximately}, the number of symbolic states obtained via the proposed approach in this paper is smaller due to the following reasons: (i) the state space sampling parameter is a constant in \cite{Girard2012controller, Pola2008approximately} due to the uniform quantization, whereas varies with the quantization parameter here; (ii) the chosen region is a subset of the state set, which leads to a smaller $\theta_{C}$ due to its proportionality to the volume of $C$. Hence, the number of symbolic states is $\sum_{i\in\mathbb{N}}\theta_{i}(\Lambda\mu_{i})^{-n}$ with the volume-related coefficient $\theta_{i}\in\mathbb{R}^{+}$ for the region $\mathcal{S}_{i}$.

\section{Illustrative Example}
\label{sec-example}

Consider an autonomous vehicle, whose dynamics is assumed to be the bicycle model; see \cite[Chapter 2.4]{Astrom2010feedback}. The dynamics is of the form \eqref{eqn-1} with $f: \mathbb{R}^{3}\times U\rightarrow\mathbb{R}^{3}$ given by
\begin{equation}
\label{eqn-19}
f (x, u)=\begin{bmatrix}
u_{1}\cos(\alpha+x_{3})\cos^{-1}(\alpha) \\
u_{1}\sin(\alpha+x_{3})\cos^{-1}(\alpha) \\
u_{1}\tan(u_{2})
\end{bmatrix},
\end{equation}
where $x:=(x_{1}, x_{2}, x_{3})\in\mathbb{R}^{3}$ is the system state with the position $(x_{1}, x_{2})$ and the orientation $x_{3}$, $u=(u_{1}, u_{2})$ is the control input with the rear wheel velocity $u_{1}$ and the steering angle $u_{2}$, and $\alpha:=\arctan(\tan(u_{2})/2)$. In addition, the state set is limited in the set $X=[0, 10]\times[0, 10]\times[-\pi, \pi]$, and the control input is given by the set $U:=[-1, 1]\times[-1, 1]$.

The control problem to be studied is formulated with respect to the sampled system $T_{\tau}(\Sigma)$ associated with \eqref{eqn-1} and the sampling time $\tau=0.3$. Here, our objective is to design a controller such that certain patrolling behavior is enforced on the vehicle in a given complex environment. The vehicle is initialized at the state $A_{1,0}=\{(0.4, 0.4, 0)\}$, needs to avoid the obstacles $A_{1, a}$ (that is, the black regions in Fig. \ref{fig-2}) and to patrol infinitely often between two target regions $A_{1, r_{1}}=[0, 0.5]\times[0, 0.5]\times\mathbb{R}$ and $A_{1, r_{2}}=(9, 0, 0)+A_{1, r_{1}}$. For the obstacle regions, the third component of $A_{1, a}$ equals to $\mathbb{R}$. Therefore, the specification $\mathcal{O}^{1}_{s}$ is defined as: $\{(u, x)\in(U_{1}\times X_{1})^{\mathbb{N}^{+}}: x(0)\in A_{1, 0}\Rightarrow\forall_{t\in\mathbb{N}^{+}}(x(t)\notin A_{1, a}\wedge\forall_{i\in\{1, 2\}}\exists_{t'\in[t, \infty)}x(t')\in A_{1, r_{i}})\}$, where $U_{1}=U$ and $X_{1}=\mathbb{R}^{3}$. Hence, the control problem is to design a control strategy for $T_{\tau}(\Sigma)$ such that the specification $\mathcal{O}^{1}_{s}$ is satisfied.

To deal with such a control problem, we aim to solve its abstract version based on the developed results in the previous sections. To this end, we first construct the symbolic abstraction for $T_{\tau}(\Sigma)$. According to the initial state $A_{1,0}$, we choose the initial region $\mathcal{S}_{0}:=[0, 0.6]\times[0, 0.6]\times[-4\pi/35, 4\pi/35]$ and the constant $\omega=0.1$. For the applied dynamic quantizer, let $(\Lambda_{1}, \Lambda_{2}, \Lambda_{3})=(0.2, 0.2, 2\pi/35)$, $\mu=1$ and $\Omega_{\ini}=\Omega_{\out}=1$. That is, all the chosen regions are the same and quantized uniformly. The quantized states are given by $\imath(\Lambda_{1}, \Lambda_{2})\times\jmath\Lambda_{3}$ with $\imath\in\mathbb{N}$ and $\jmath\in\pm\mathbb{N}$. There are 80 abstract states in $\mathcal{S}_{0}$. Given a precision $\varepsilon=0.2$, we choose $\eta_{0}=0.2$, and then construct a local symbolic abstraction $T_{\tau, \mu, \eta_{0}}(\Sigma)$ for the initial region $\mathcal{S}_{0}$. Along the same mechanism and the chosen region, the local symbolic abstraction can be constructed for all the chosen regions.

\begin{figure}
\centering
\subfigure
{
\label{fig-2-a}
\begin{picture}(60, 92)
\put(-55,-7){\resizebox{60mm}{33mm}{\includegraphics[width=2in]{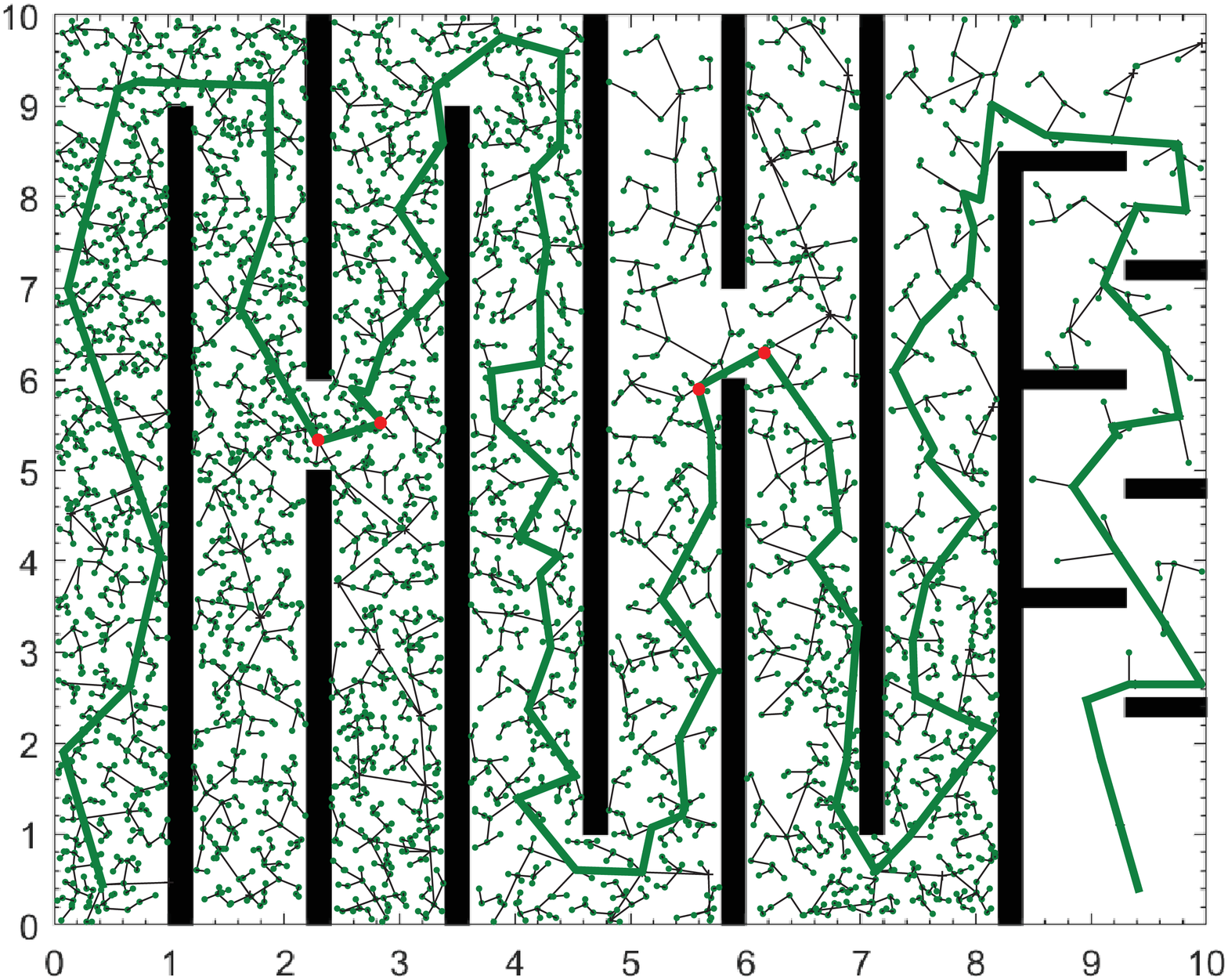}}}
\end{picture}
}

\subfigure
{
\label{fig-2-b}
\begin{picture}(60, 92)
\put(-55,-7){\resizebox{60mm}{33mm}{\includegraphics[width=2in]{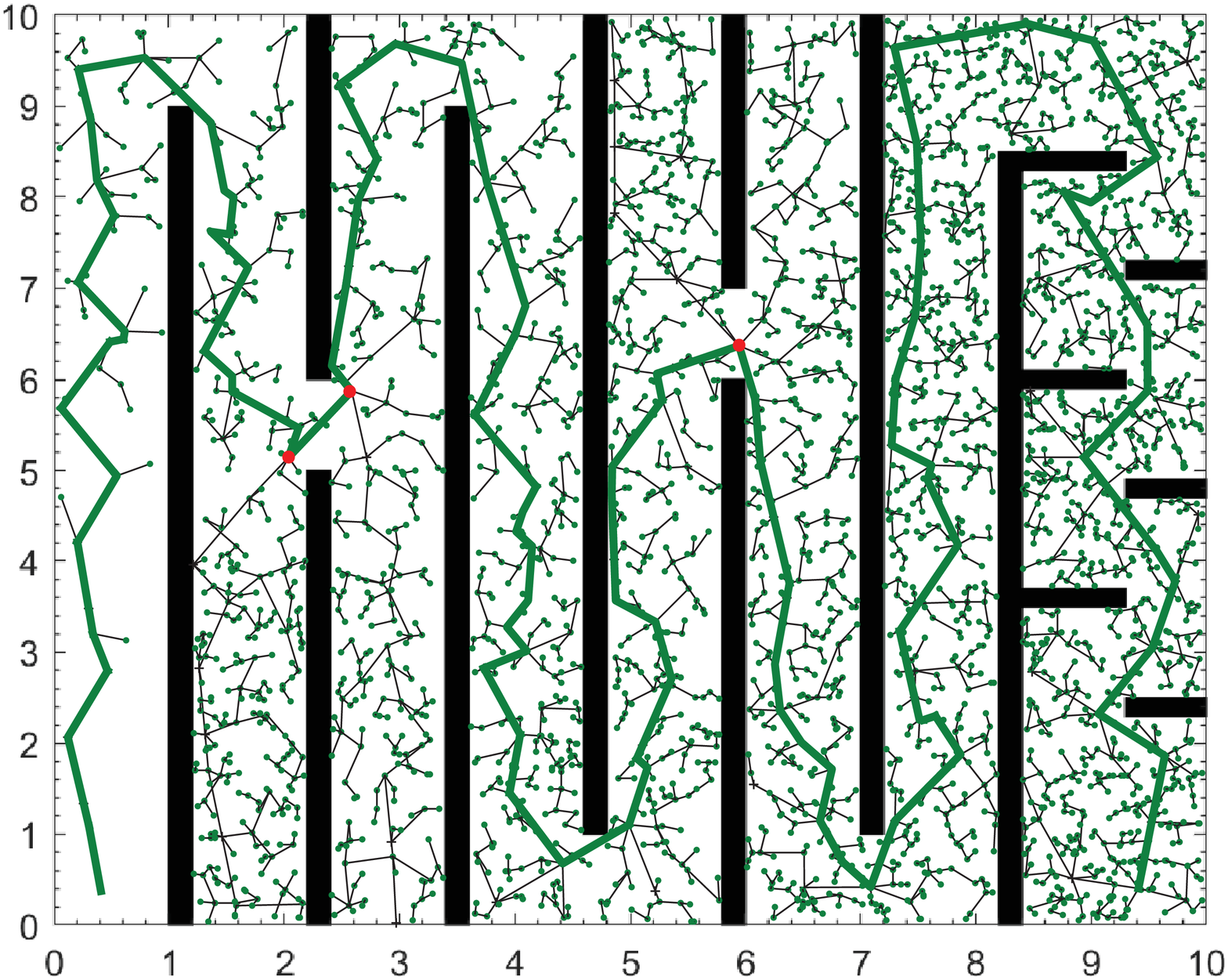}}}
\end{picture}
}
\caption{Projection of the states of $\Sigma$ and $T_{\tau, \bar{\mu}, \bar{\eta}}(\Sigma)$ to $\mathbb{R}^{2}\times\{0\}$. The black regions are the obstacles, and the green lines are the optimal path planning strategies, which is searched via the RRT algorithm. The upper figure is for the path planning from $A_{1, r_{1}}$ to $A_{1, r_{2}}$, and the lower figure is for the path planning from $A_{1, r_{2}}$ to $A_{1, r_{1}}$.}
\label{fig-2}
\end{figure}

Using the RRT algorithm in \cite{Karaman2011sampling}, the path planning of the autonomous vehicle is presented in Fig. \ref{fig-2} and the number of the chosen regions is 103. As a result, there are 8240 abstract states in all the chosen regions, whereas there are 91035 abstract states in \cite{Zamani2012symbolic}, which implies that the computation time for the abstraction are reduced greatly. Note that the designed controller is optimal in each chosen region; see Fig. \ref{fig-2}. In the chosen regions, all the possible paths are searched. A possible case is that the designed controller is not optimal after finite chosen regions. In this case, we can choose the suboptimal path to continue the searching process until the vehicle reaches into the target regions. We need to emphasize that, different path planning strategies have great effects on the satisfaction of the desired specification. For instance, if another path planning strategy is chosen at certain abstract states (e.g., the red dots in Fig. \ref{fig-2}), then the longer computation time is needed.

\section{Conclusion}
\label{sec-conclusion}

In this paper, we proposed a dynamic abstraction for nonlinear control systems via dynamic quantization. To reduce computational complexity in the construction of symbolic abstractions, a zoom quantizer with bounded quantization regions was applied. Using the zoom quantizer, a dynamic approximation approach was proposed for the state and input sets. Based on this dynamic approximation, the dynamic symbolic abstractions were constructed. Finally, a numerical example was presented to illustrate the derived results.



\begin{thebibliography}{10}

\bibitem{Milner1989communication}
R.~Milner, \emph{Communication and Concurrency}.\hskip 1em plus 0.5em minus
  0.4em\relax Prentice Hall, 1989.

\bibitem{Tabuada2006linear}
P.~Tabuada and G.~J. Pappas, ``Linear time logic control of discrete-time
  linear systems,'' \emph{IEEE Trans. Autom. Control}, vol.~51, no.~12, 2006.

\bibitem{Ramadge1987supervisory}
P.~J. Ramadge and W.~M. Wonham, ``Supervisory control of a class of discrete
  event processes,'' \emph{SIAM J. Control Optim.}, vol.~25, no.~1, pp.
  206--230, 1987.

\bibitem{Ramadge1987modular}
P.~J. Ramadge and W.~M. Wonham, ``Modular feedback logic for discrete event systems,'' \emph{SIAM J.
  Control Optim.}, vol.~25, no.~5, pp. 1202--1218, 1987.

\bibitem{Girard2012controller}
A.~Girard, ``Controller synthesis for safety and reachability via approximate
  bisimulation,'' \emph{Automatica}, vol.~48, no.~5, pp. 947--953, 2012.

\bibitem{Pola2008approximately}
G.~Pola, A.~Girard, and P.~Tabuada, ``Approximately bisimilar symbolic models
  for nonlinear control systems,'' \emph{Automatica}, vol.~44, no.~10, pp.
  2508--2516, 2008.

\bibitem{Girard2010approximately}
A.~Girard, G.~Pola, and P.~Tabuada, ``Approximately bisimilar symbolic models
  for incrementally stable switched systems,'' \emph{IEEE Trans. Autom.
  Control}, vol.~55, no.~1, pp. 116--126, 2010.

\bibitem{Girard2007approximation}
A.~Girard and G.~J. Pappas, ``Approximation metrics for discrete and continuous
  systems,'' \emph{IEEE Trans. Autom. Control}, vol.~5, no.~52, pp. 782--798,
  2007.

\bibitem{Ren2018quantized}
W.~Ren and J.~Xiong, ``Quantized feedback stabilization of nonlinear systems
  with external disturbance,'' \emph{IEEE Trans. Autom. Control}, vol.~63,
  no.~9, pp. 3167--3172, 2018.

\bibitem{Delchamps1990stabilizing}
D.~F. Delchamps, ``Stabilizing a linear system with quantized state feedback,''
  \emph{IEEE Trans. Autom. Control}, vol.~35, no.~8, pp. 916--924, 1990.

\bibitem{Liberzon2003hybrid}
D.~Liberzon, ``Hybrid feedback stabilization of systems with quantized
  signals,'' \emph{Automatica}, vol.~39, no.~9, pp. 1543--1554, 2003.

\bibitem{Zamani2012symbolic}
M.~Zamani, G.~Pola, M.~Mazo, and P.~Tabuada, ``Symbolic models for nonlinear
  control systems without stability assumptions,'' \emph{IEEE Trans. Autom.
  Control}, vol.~57, no.~7, pp. 1804--1809, 2012.

\bibitem{Sontag2013mathematical}
E.~D. Sontag, \emph{Mathematical Control Theory: Deterministic Finite
  Dimensional Systems}.\hskip 1em plus 0.5em minus 0.4em\relax Springer Science
  \& Business Media, 1998.

\bibitem{Angeli1999forward}
D.~Angeli and E.~D. Sontag, ``Forward completeness, unboundedness
  observability, and their {L}yapunov characterizations,'' \emph{Systems \&
  Control Letters}, vol.~38, no. 4-5, pp. 209--217, 1999.

\bibitem{Zamani2015symbolic}
M.~Zamani, A.~Abate, and A.~Girard, ``Symbolic models for stochastic switched
  systems: A discretization and a discretization-free approach,''
  \emph{Automatica}, vol.~55, pp. 183--196, 2015.

\bibitem{Elia2001stabilization}
N.~Elia and S.~K. Mitter, ``Stabilization of linear systems with limited
  information,'' \emph{IEEE Trans. Autom. Control}, vol.~46, no.~9, pp.
  1384--1400, 2001.

\bibitem{Coutinho2010input}
D.~F. Coutinho, M.~Fu, and C.~E. de~Souza, ``Input and output quantized
  feedback linear systems,'' \emph{IEEE Trans. Autom. Control}, vol.~55, no.~3,
  pp. 761--766, 2010.

\bibitem{Heemels2010networked}
W.~M.~H. Heemels, A.~R. Teel, N.~van~de Wouw, and D.~Ne{\v{s}}i{\'c},
  ``Networked control systems with communication constraints: {T}radeoffs
  between transmission intervals, delays and performance,'' \emph{IEEE Trans.
  Autom. Control}, vol.~55, no.~8, pp. 1781--1796, 2010.

\bibitem{Liberzon2003switching}
D.~Liberzon, \emph{Switching in Systems and Control}.\hskip 1em plus 0.5em
  minus 0.4em\relax Springer Science \& Business Media, 2003.

\bibitem{Forstner2002discrete}
D.~F{\"o}rstner, M.~Jung, and J.~Lunze, ``A discrete-event model of
  asynchronous quantised systems,'' \emph{Automatica}, vol.~38, no.~8, pp.
  1277--1286, 2002.

\bibitem{Hashimoto2017robust}
K.~Hashimoto, S.~Adachi, and D.~V. Dimarogonas, ``Robust safety controller
  synthesis using tubes,'' in \emph{IEEE Conference on Decision and
  Control}.\hskip 1em plus 0.5em minus 0.4em\relax IEEE, 2017, pp. 535--541.

\bibitem{Girard2013low}
A.~Girard, ``Low-complexity quantized switching controllers using approximate
  bisimulation,'' \emph{Nonlinear Analysis: Hybrid Systems}, vol.~10, pp.
  34--44, 2013.

\bibitem{Girard2016safety}
A.~Girard, G.~G{\"o}ssler, and S.~Mouelhi, ``Safety controller synthesis for
  incrementally stable switched systems using multiscale symbolic models,''
  \emph{IEEE Trans. Autom. Control}, vol.~61, no.~6, pp. 1537--1549, 2016.

\bibitem{Saoud2017optimal}
A.~Saoud and A.~Girard, ``Optimal multirate sampling in symbolic models for
  incrementally stable switched systems,'' \emph{Automatica}, 2017.

\bibitem{Astrom2010feedback}
K.~J. Astr{\"o}m and R.~M. Murray, \emph{Feedback Systems: An Introduction for
  Scientists and Engineers}.\hskip 1em plus 0.5em minus 0.4em\relax Princeton
  University Press, 2010.

\bibitem{Karaman2011sampling}
S.~Karaman and E.~Frazzoli, ``Sampling-based algorithms for optimal motion
  planning,'' \emph{The International Journal of Robotics Research}, vol.~30,
  no.~7, pp. 846--894, 2011.

\end{thebibliography}
\end{document}